\documentclass[letterpaper, 10 pt, conference]{ieeeconf}
\IEEEoverridecommandlockouts

\usepackage{latexsym}
\usepackage{graphicx}
\usepackage{amsfonts,amssymb,amsmath}
\usepackage{hyperref}
\usepackage{cleveref}
\usepackage{ mathrsfs}
\usepackage{easyReview}
\usepackage{cite}

\newtheorem{lemma}{Lemma}
\newtheorem{theorem}{Theorem}

\title{\bf
Rigid Body Geometric Attitude Estimator using Multi-rate Sensors}

\author{Maulik Bhatt$^{1}$\thanks{$^{1}$Undergraduate Student, Aerospace Engineering, Indian Institute of Technology Bombay, 400076, India. {\tt \footnotesize maulik.bhatt@iitb.ac.in}}, Srikant Sukumar$^{2}$\thanks{$^{2}$Associate Professor , Systems and Control Engineering, Indian Institute of Technology Bombay, 400076, India. {\tt \footnotesize srikant@sc.iitb.ac.in} }, Amit K. Sanyal$^{3}$\thanks{$^{3}$Associate Professor,Mechanical and Aerospace Engineering, Syracuse University, Syracuse, NY, USA. {\tt \footnotesize aksanyal@syr.edu}}
}

\begin{document}

\maketitle

\begin{abstract}
    A geometric estimator is proposed for the rigid body attitude under multi-rate measurements using discrete-time Lyapunov stability analysis in this work. The angular velocity measurements are assumed to be sampled at a higher rate compared to the attitude. The attitude determination problem from two or more vector measurements in the body-fixed frame is formulated as Wahba's problem. In the case when measurements are absent, a discrete-time model for attitude kinematics is assumed in order to propagate the measurements. A discrete-time Lyapunov function is constructed as the sum of a kinetic energy-like term that is quadratic in the angular velocity estimation error and an artificial potential energy-like term obtained from Wahba's cost function. A filtering scheme is obtained by discrete-time stability analysis using a suitable Lyapunov function. The analysis shows that the filtering scheme is exponentially stable in the absence of measurement noise and the domain of convergence is almost global. For a realistic evaluation of the scheme, numerical experiments are conducted with inputs corrupted by bounded measurement noise. Simulation results exhibit convergence of the estimated states to a bounded neighborhood of the actual states.
\end{abstract}


\section{Introduction}\label{sec:1}
Attitude estimation of rigid bodies finds a wide variety of applications including spacecrafts, robotics, underwater vehicles, aerial vehicles and so on. In this work, we address the estimation problem for attitude and angular velocity of a rigid body given multi-rate measurements. Attitude estimators typically rely on two kinds of measurements in the body-fixed frame, 1) known inertial vector measurements, and 2) angular velocity measurements. In practice, however, these two measurements may not be available at the same time. The number of observed inertial directions may also vary over time. However, it is assumed that the number of observed inertial directions are at least two so that the attitude can be uniquely determined from the measured directions.

One of the earliest solutions to such a problem is found in \cite{black1964passive} where the TRIAD algorithm is used to determine the rotation matrix using two independent vector measurements. The limitation of this algorithm is its sensitivity to noise. In the further developments, perhaps the most influential work in the field of attitude estimation field was proposed by Wahba, as an optimization problem for estimating the attitude by minimizing the sum of the squared norms of vector errors using three or more vector measurements, in \cite{wahba1965least}. Solutions to the Wahba's problem have been attempted via multiple methods. Markley solved it using the Singular Value Decomposition (SVD) method in \cite{markley1988attitude}. QUEST algorithm, which determines the attitude that achieves the best-weighted overlap of an arbitrary number of reference vectors, is presented in \cite{shuster1981three}. In \cite{mortari1997esoq}, Mortari presented the EStimator of the Optimal Quaternion (ESOQ) algorithm, which provides the closed-form expressions of a $4\times4$ matrix's eigenvalues and then computes the eigenvector associated with the greatest of them, representing the optimal quaternion. Numerical solutions to the Wahba's problem are presented in \cite{psiaki2012numerical}. 

Comprehensive surveys of various filtering based methods employed in attitude determination are available in \cite{crassidis2007survey,madinehi2013rigid}. However, most of them either present the attitude estimation scheme in continuous-time or neglect the delay in the attitude measurements. One of the earliest attempts to solve the problem of rigid body attitude estimation with multi-rate measurements is found in \cite{sanyal2012attitude} using uncertainty ellipsoids. \cite{khosravian2015recursive} presents a recursive method based on the cascade combination of an output predictor and an attitude observer. Attitude estimation using single delayed vector measurement and biased gyro appeared in \cite{bahrami2014delay, bahrami2017global}. Velocity aided attitude estimation with sensor delay is presented in \cite{khosravian2014velocity}. Nonlinear complementary filters for rigid body attitude estimation are presented in \cite{mahony2008nonlinear}. Few other examples where non-linear or geometric methods used in determining attitude estimates are \cite{vasconcelos2007landmark,vasconcelos2008nonlinear,valpiani2008nonlinear}. However, \cite{mahony2008nonlinear,vasconcelos2007landmark,vasconcelos2008nonlinear,valpiani2008nonlinear} do not address the multi-rate measurement case.

As evident from above, the problem of attitude estimation in case of multi-rate measurements in discrete-time without any assumptions on the measurement noise and number of observed directions has not been addressed in a geometric framework. The focus of the current work is, therefore, the development of a geometric attitude determination scheme under multi-rate measurements in discrete time, with robustness to noise guarantees. In the geometric approach, the attitude is represented globally via the rotation matrix without using local coordinates. We do not assume any specific statistics on the measurement noise (such as noise distribution, variance, etc.) but that it is bounded. The multi-rate discrete-time filtering scheme presented here is obtained by using the discrete-Lyapunov method applied on a Lyapunov candidate that depends on the state estimation errors. The filtering scheme provided is asymptotically stable with almost global convergence. In \cite{izadi2014rigid}, a filtering scheme in continuous-time is proposed by applying the Lagrange-d'Alembert principle on suitably formulated artificial kinetic and potential energy functions. In \cite{izadi2014rigid}, the authors formulate filter equations assuming that inertial vector measurements and angular velocity measurements are available synchronously and continuously. We relax that assumption in this article and provide a provably stable geometric filter for attitude estimation under multi-rate measurements. 

This paper is organized as follows. In \Cref{sec:2}, the attitude estimation problem is formulated as Wahba's optimization problem and then some important properties of the Wahba's cost function are presented. In the \Cref{sec:3}, the propagation model for the measurements in the multi-rate measurement case is presented and then an exponentially stable discrete-time estimator with an almost global domain of convergence is derived using the discrete-time Lyapunov method. The domain of convergence is identical to that shown in \cite{izadi2014rigid}. Filter equations are numerically verified with realistic measurements (corrupted by bounded noise) in \Cref{sec:4}. Finally, \Cref{sec:5} presents the concluding remarks with contributions and future work.

\section{Attitude determination from vector measurements}\label{sec:2}

Rigid body attitude can be determined by measuring $k \in \mathbb{N}$ known and independent inertial vectors in the body-fixed frame. Let these vectors in the coordinate frame fixed to the body be denoted by $u_i^m$ for $i = 1,\ldots,k$, where $k \geq 2$. The assumption of $k\geq2$ is required for the unique determination of attitude at a particular instant. For $k=2$, the cross product of the two measured vectors is used as the third independent measurement for attitude determination. Let the corresponding known inertial vectors be denoted by $e_i$. Also, let the true vectors in the body-fixed frame be denoted by $u_i := R^Te_i$, where $R$ is the rotation matrix of the body-fixed frame with respect to the inertial frame. This rotation matrix provides a coordinate-free global and unique description of the attitude of the rigid body. Define the matrix composed of all $k$ measured vectors expressed in the body-fixed frame as column vectors,
\begin{align}\label{eq:1}
    U^m & = [\begin{matrix} u_1^{m} & u_2^{m} & u_1^{m}\times u_2^{m}\end{matrix}] \in \mathbb{R}^{3\times3} \; \text{when} \; k=2 \; \text{and}, \nonumber\\
U^m & = [\begin{matrix} u_1^{m} & u_2^{m} & \ldots &u_k^{m}\end{matrix}] \in \mathbb{R}^{3\times k} \; \text{when} \; k > 2
\end{align}
and expressing them in inertial frame,
\begin{align}\label{eq:2}
    E & = [\begin{matrix} e_1 & e_2 & e_1\times e_2\end{matrix}] \in \mathbb{R}^{3\times3} \; \text{when} \; k=2 \;\text{and}, \nonumber\\
E & = [\begin{matrix} e_1 & e_2 & \ldots &e_k\end{matrix}] \in \mathbb{R}^{3\times k} \; \text{when} \; k > 2
\end{align}
The \emph{true} body vector matrix is as below.
\begin{align}\label{eq:3}
    U &= R^{T}E = [\begin{matrix} u_1 & u_2 & u_1\times u_2\end{matrix}] \in \mathbb{R}^{3\times3} \; \text{when} \;  k=2 \; \text{and}, \nonumber\\
U & = R^{T}E = [\begin{matrix} u_1 & u_2 & \ldots &u_k\end{matrix}] \in \mathbb{R}^{3\times k} \; \text{when} \; k > 2
\end{align}

\subsection{Generalization of Wahba's cost function for instantaneous attitude determination from vector measurements}

The optimal attitude determination problem using a set of vector measurements is finding an estimated rotation matrix $\hat{R} \in SO(3)$, where $SO(3) := \{R \in \mathbb{R}^{3\times3} \; | \; R^TR = RR^T = I\}$, such that a weighted sum of squared norms of the vector errors,
\begin{equation}
    s_i = e_i - \hat{R}u_i^m
\end{equation}
is minimized. This attitude determination problem is known as Wahba's problem and consists of minimizing the value of 
\begin{equation}\label{eq:5}
     \mathcal{U}^{0}(\hat{R},U^{m}) = \frac{1}{2}\sum_{i=1}^{k}w_i (e_i - \hat{R}u_i^m)^T(e_i - \hat{R}u_i^m)
\end{equation}
with the respect to $\hat{R} \in SO(3)$, where the weights $w_i > 0$ for all $i \in \{1,2,\ldots,k\}$. Defining the trace inner product on $\mathbb{R}^{m\times n}$ as
\begin{equation}
    \langle\,A_1,A_2\rangle := \text{trace}(A_1^{T}A_2)
\end{equation}
we can express \cref{eq:5} as,
\begin{equation}\label{eq:7}
    \mathcal{U}^0(\hat{R},U^{m}) = \frac{1}{2} \langle\,E - \hat{R}U^m,(E-\hat{R}U^m)W\rangle
\end{equation}
where $U^m$ is given by \cref{eq:1}, $E$ is given by \cref{eq:2}, and $W = \text{diag}(w_i)$ is the positive definite diagonal matrix of the weight factors for the measured directions. \\

$W$ in \cref{eq:7} can be generalized to be any positive definite matrix. Another generalization of Wahba's cost function is given by,
\begin{equation}\label{eq:8}
    \mathcal{U}(\hat{R},U^{m}) = \mathit{\Phi}\left(\frac{1}{2} \langle\,E - \hat{R}U^m,(E-\hat{R}U^m)W\rangle\right)
\end{equation}
where, $\mathit{\Phi}:[0,\infty) \mapsto [0,\infty)$ is a $\mathcal{C}^2$ function with $\mathit{\Phi}(0)=0$ and $\mathit{\Phi^{\prime}(x)}>0, \; \forall x \in [0,\infty)$. Further, $\mathit{\Phi^{\prime}(x) \leq \alpha(x)}$ where, $\alpha(\cdot)$ is a class $\mathcal{K}$-function. These properties of $\mathit{\Phi}(\cdot)$ ensure that $\mathcal{U}^0(\hat{R},U^m)$ and $\mathcal{U}(\hat{R},U^m)$ have the same minimizer $\hat{R}^*\in SO(3)$. In other words, minimizing the cost $\mathcal{U}$, which is a generalization of the cost $\mathcal{U}^0$, is equivalent to solving Wahba's problem. Here, $W$ is symmetric positive definite, $E$ and $U^m$ are assumed to be of rank 3 which is consistent with assuming that $k\geq2$ measurements are available. 

\subsection{Properties of Wahba's cost function in the absence of measurements errors}

For the case of zero measurement errors (noise), we have $U^m = U = R^TE$. Let $Q = R\hat{R}^T \in SO(3)$ denote the attitude estimation error. Let $(\cdot)^\times:\mathbb{R}^3 \mapsto \mathfrak{so}(3) \subset \mathbb{R}^{3\times3}$ be the skew-symmetric matrix cross-product operator and denotes the vector space isomorphism between $\mathbb{R}^3$ and $\mathfrak{so}(3)$, where $\mathfrak{so}(3) := \{ M \in \mathbb{R}^{3\times3} \; | \; M + M^T = 0\}$:

\begin{equation}
    v^\times = {\begin{bmatrix} v_1 \\ v_2 \\ v_3 \end{bmatrix}}^{\times} = \begin{bmatrix} 0 & -v_3 & v_2 \\ v_3 & 0 & -v_1\\ -v_2 & v_1 & 0 \end{bmatrix}
\end{equation}
Further, let $\text{vex}(\cdot):\mathfrak{so}(3)\mapsto\mathbb{R}^3$ be the inverse of $(\cdot)^\times$.
The following lemmas from \cite{izadi2014rigid} stated here without proof give the structure and characterization of critical points of the Wahba's cost function. \\

\begin{lemma}\label{lemma:1}
Let rank($E$) = 3 and the singular value decomposition of $E$ be given by,
\begin{align}
    & E := U_E\Sigma_EV_E^T \; \text{where} \; U_E \in O(3), V_E \in SO(m). \nonumber\\
 & \Sigma_E \in  Diag^+(3,m),
\end{align}
and Diag$^+(n_1,n_2)$ is the vector space of $n_1\times n_2$ matrices with positive entries along the main diagonal and all the other components zero. Let $\sigma_1,\sigma_2,\sigma_3$ denote the main diagonal entries of $\Sigma_E$. Further, Let $W$ from \cref{eq:7} be given by,
\begin{equation}
    W = V_EW_0V_E^T \; \text{where} \; W_0 \in \text{Diag}^+(m,m)
\end{equation}
and the first three diagonal entries of $W_0$ are given by,
\begin{equation}
    w_1 = \frac{d_1}{\sigma_1^2}, \;  w_2 = \frac{d_2}{\sigma_2^2}, \;  w_3 = \frac{d_3}{\sigma_3^2} \;\; \text{where} \; d_1,d_2,d_3 > 0
\end{equation}
Then, $K = EWE^T$ is positive definite and,\\
\begin{equation}
    K = U_E\Delta U_E^T \; \text{where} \; \Delta = \text{diag}(d_1,d_2,d_3)
\end{equation}
is its eigen decomposition. Moreover, if $d_i \neq d_j$ for $i \neq j$ and $i,j \in \{1,2,3\}$ then $\langle\,I-Q,K\rangle$ is a Morse function whose set of critical points given as the solution of $S_K(Q) := \text{vex}\left(KQ^T - QK\right) = 0$ are,
\begin{equation}\label{eq:13}
    C_Q := \{I,Q^1,Q^2,Q^3\} \; \text{where} \; Q^i = 2U_Ea_ia_i^TU_E^T - I
\end{equation}
and $a_i$ is the $i^{th}$ column vector of the identity matrix $I \in SO(3)$.
\end{lemma}
\begin{lemma}\label{lemma:2}
Let $K =EWE^T$ have the properties given by \Cref{lemma:1}. Then the map $\mathit{\Phi}: SO(3) \rightarrow \mathbb{R}$, $Q \mapsto \Phi(\langle\,I-Q,K\rangle)$  with critical points given by \cref{eq:13} has a global minimum at the identity $I\in SO(3)$, a global maximum and two hyperbolic saddle points whose indices depend on the distinct eigenvalues $d_1,d_2,$ and $d_3$ of $K$. 
\end{lemma}

\section{Discrete-time estimator in the presence of multirate measurements}\label{sec:3}

\subsection{Discretization of Attitude Kinematics}

Consider the time interval $[t_0,T] \subseteq \mathbb{R}^+$ divided into N equal sub-intervals $[t_i,t_{i+1}]$ for $i = 0,1,\dots,N$ with $t_N = T $ and let $t_{i+1} - t_i = h$ be the time step size. Let the true angular velocity in the body-fixed frame be denoted by $\Omega \in \mathbb{R}^3$. The true and measured angular velocities at the time instant $t_i$ will be denoted by $\Omega_i$ and $\Omega_i^m$ respectively. Further, let $U_i$ and $U_{i}^m$ denote the matrix formed by true and measured inertial vectors in the body-fixed frame at the time instant $t_i$ respectively. The assumption is that angular velocity measurements and inertial vectors measurements in the body-fixed frame are coming at a different but constant rate. In general coarse rate gyros have much higher sampling rate than that of a coarse attitude sensor. Therefore, in a realistic scenario, angular velocities are measured at a higher rate than the inertial vector measurements in the body-fixed frame. Therefore, we assume that the measurements of angular velocity ($\Omega^m$) are available after each time interval $h$ say, $\Omega^m_0,\Omega^m_1,\dots,\Omega^m_N$ while, inertial vector measurements in the body-fixed frame are available after time interval $nh, n \in N$ say, $U^m_0,U^m_n,U^m_{2n},\dots$.\\

We have, $U = R^TE$. Therefore, at time instants $t_i$ and $t_{i+1}$, the following relations will hold true respectively; $U_i = R_i^TE_i, \; U_{i+1} = R_{i+1}^TE_{i+1}$. Here, $R_i$ and $R_{i+1}$ are the rotation matrices from body-fixed frame to inertial frame at time instants $t_i$ and $t_{i+1}$ respectively. $E_i = E_{i+1} = E$ are the corresponding known vectors expressed in the inertial frame. Note that the vectors are fixed in the inertial frame and do not change with the time. \\

The continuous time attitude kinematics are,
\begin{equation}\label{eq:15}
    \Dot{R} = R\Omega^\times
\end{equation}
We discretize the kinematics in \cref{eq:15} as follows,
\begin{equation}\label{eq:14}
    R_{i+1} = R_i\exp{\left(\frac{h}{2}(\Omega_{i+1}+\Omega_i)^\times\right)}
\end{equation}
where, $\exp{(\cdot)}:\mathfrak{so}(3)\mapsto SO(3)$ is the map defined as,
\begin{equation}
    \exp{(M)} = \sum_{i=0}^{\infty}\frac{1}{k!}M^k
\end{equation}

Using \cref{eq:3} and the discretization from \cref{eq:14},
\begin{align}\label{eq:16}
    U_{i+1} & = \exp{\left(-\frac{h}{2}(\Omega_{i+1}+\Omega_i)^\times\right)}R_i^TE_i \nonumber\\ & = \exp{\left(-\frac{h}{2}(\Omega_{i+1}+\Omega_i)^\times\right)}U_i
\end{align}

For the instants of time when inertial vector measurements in the body-fixed frame are not available we will use \cref{eq:16} to obtain the missing values of $U_i^m$. This implies that for the time instants $ (n-1)h < t_i < nh, n \in \mathbb{N}$, by employing the propagation scheme in \cref{eq:16}, we propagate direction vector measurements between the instants at which they are measured, using the angular velocity measurements that are obtained at a faster rate. We now formalise the aforementioned inertial vector measurement model as below,
\begin{equation}\label{eq:17}
    \Tilde{U}_i^m := \begin{cases}
        U_i^m, & \text{if} \; i\,mod\,n = 0
        \\
        \exp{\left(-\frac{h}{2}(\Omega_{i-1}^m+\Omega_i^m)^\times\right)}\Tilde{U}^m_{i-1}, & \text{otherwise}.
    \end{cases}
\end{equation}

Note that in the absence of measurements errors, we have $\Omega_i^m = \Omega_i,\, \forall i \in \{0,1,\ldots,N\}$.  Also, $U_i^m = U_i$ for the time instants when inertial vector measurements are available. Now, at time instant $t_0$, we have $\Tilde{U}^m_{0}=U^m_0 = U_0$ and $\Omega^m_0 = \Omega_0$. Using \cref{eq:17} at time  instant $t_1$, noting that $\Omega^m_1 = \Omega_1$, we get $\Tilde{U}^m_{1} = \exp{\left(-\frac{h}{2}(\Omega_{0}+\Omega_1)^\times\right)}{U}_{0}$. Comparing it with \cref{eq:16}, we have $\Tilde{U}^m_{1} = U_1$. Using the relation from \cref{eq:3} we have $\Tilde{U}^m_{1} = R_1^TE_1$. Similarly, combining \cref{eq:16}, and \cref{eq:17}, and using the relation in \cref{eq:3} we get the following relation for all $i \in \{0,1,\ldots,N\}$ in the absence of measurement errors.
\begin{equation}
    \Tilde{U}_i^m = R_i^TE_i
\end{equation}

\subsection{Discrete-time attitude state estimation using the discrete Lyapunov Approach}

The value of the Wahba's cost function at each instant encapsulates the error in the attitude estimation. We can consider the Wahba's cost function as an artificial potential energy-like term. Therefore using \cref{eq:8} we have,
\begin{equation}\label{eq:18}
    \mathcal{U}(\hat{R}_i,\Tilde{U}^{m}_i) = \mathit{\Phi}\left(\frac{1}{2} \langle\,E_i - \hat{R}_i\tilde{U}^m_i,(E_i-\hat{R}_i\Tilde{U}_i^m)W_i\rangle\right)
\end{equation}

The term encapsulating the "energy" in the angular velocity estimation error is denoted by the map $\mathcal{T}:\mathbb{R}^3\times\mathbb{R}^3\mapsto \mathbb{R}$ defined as,
\begin{equation} \label{eq:19}
    \mathcal{T} (\hat{\Omega}_i,\Omega^m_i)  = \frac{m}{2}(\Omega^m_i - \hat{\Omega}_i)^T(\Omega^m_i - \hat{\Omega}_i)
\end{equation}
where $m>0$ is a scalar and $\Tilde{U}_i^m$ is according to \cref{eq:17}. In the absence of measurement errors, we have $\Tilde{U}_i^m = R_i^TE_i$. Therefore we can we can write \cref{eq:18} in terms of state estimation error $Q_i = R_i\hat{R}^T_i$ as follows,
\begin{align}
     \mathcal{U}(\hat{R}_i,\Tilde{U}^{m}_i) & = \mathit{\Phi}\left(\frac{1}{2} \langle\,E_i - \hat{R}_iR_i^TE_i,(E_i-\hat{R}_iR_i^TE_i)W_i\rangle\right) \nonumber \\
     & = \mathit{\Phi}\left(\langle\,I - R_i\hat{R}_i^T,E_iW_iE_i^T\rangle\right) \nonumber \\
     \Rightarrow \mathcal{U}(Q_i) & =  \mathit{\Phi}(\langle\,I-Q_i,K_i\rangle) \;\; \text{where} \;\; K_i = E_iW_iE_i^T
\end{align}

The weights $W_i$s are chosen such that $K_i$ is always positive definite with distinct eigenvalues according to \cref{lemma:1}. Further, \cref{eq:19} can be written in terms of angular velocity estimation error, $\omega_i := \Omega^m_i - \hat{\Omega}_i$ as follows.
\begin{equation}
    \mathcal{T} (\omega_i) = \frac{m}{2}(\omega_i)^T(\omega_i)
\end{equation}
\noindent
\begin{theorem}
Consider a multi-rate measurement model for rigid body attitude determination with angular velocity available after each time interval $h>0$ denoted as, $\Omega^m_0,\Omega^m_1,\dots,\Omega^m_N$ and inertial vector measurements in the body-fixed frame being available after time interval $nh, n \in N$ denoted as, $U^m_0,U^m_n,U^m_{2n},\ldots$. Further, let the propagated inertial vector denoted by, $\Tilde{U}^m_i$ be modeled by \cref{eq:17}. Then the estimation scheme,
\begin{equation}\label{eq:22}
    \begin{cases}
    \omega_{i+1} = \frac{1}{m+l}\left[ (m-l)\omega_i + k_phS_{L_i}(\hat{R}_i) \right]\\ 
    \hat{\Omega}_i = \Omega^m_i - \omega_i \\
    \hat{R}_{i+1} = \hat{R}_i\exp{\left(\frac{h}{2}(\hat{\Omega}_{i+1}+\hat{\Omega}_i)^\times\right)}
    \end{cases}
\end{equation}
where $S_{L_{i}}(\hat{R}_i) = \text{vex}(L_i^T\hat{R}_i - \hat{R}_i^TL_i) \in \mathbb{R}^3$, $L_i = E_iW_i(\Tilde{U}^m_i)^T$, $l>0$, $l \neq m$ and $k_p > 0$, is asymptotically stable at the estimation error state $(Q,\omega) := (I,0)$ ($Q_i = R_i\hat{R}^T_i$) in the absence of measurement noise. Further, the domain of attraction of $(I,0)$ is a dense open subset of $SO(3)\times\mathbb{R}^3$.
\end{theorem}

\begin{proof}
Using the third equation from \cref{eq:22},
\begin{align} \label{eq:23}
     Q_{i+1} & = R_{i+1}\hat{R}^T_{i+1} \nonumber \\ & = Q_i\hat{R}_i\exp{\left(\frac{h}{2}(\hat{\omega}_{i+1} + \hat{\omega}_i)^\times\right)}\hat{R}_i^T
\end{align}

Let's denote,
\begin{align}
    & \mathcal{U}_i := \mathcal{U}(Q_i) = \mathit{\Phi}(\langle\,I-Q_i,K_i\rangle) \\
    & \mathcal{T}_i := \mathcal{T}(\omega_i) = \frac{m}{2}(\omega_i)^T(\omega_i)
\end{align}
  
We choose the following discrete-time Lyapunov candidate,
\begin{equation}
    V_i := V(Q_i,\omega_i) := k_p\mathcal{U}_i + \mathcal{T}_i
\end{equation}
where $k_p>0$ is a constant.

The stability of the attitude and angular velocity error can be shown by analyzing $\Delta V_i = k_p\Delta\mathcal{U}_i + \Delta\mathcal{T}_i$.\\

Assuming $\mathit{\Phi}$ to be the identity map and $K_i$ to be constant and let $K = K_i = K_{i+1}$
\begin{align}
     & \Delta\mathcal{U}_i = \mathcal{U}_{i+1} - \mathcal{U}_i = \langle\,I-Q_{i+1},K\rangle - \langle\,I-Q_i,K\rangle \nonumber \\
     & \Delta\mathcal{U}_i = \langle\,Q_i - Q_{i+1},K\rangle = -\langle\,\Delta Q_i,K\rangle
\end{align}
where, $\Delta Q_i = Q_{i+1} - Q_i$. Now,
\begin{align}
    \Delta Q_i & = Q_{i+1} - Q_i \nonumber \\
    & = Q_i\left[\hat{R}_i\exp{\left(\frac{h}{2}(\hat{\omega}_{i+1} + \hat{\omega}_i)^\times\right)}\hat{R}_i^T - I\right]
\end{align}

Approximating $\exp{\left(\frac{h}{2}(\hat{\omega}_{i+1} + \hat{\omega}_i)^\times\right)}$ by the first two terms in the expansion as,
\begin{equation}
    \exp{\left(\frac{h}{2}(\hat{\omega}_{i+1} + \hat{\omega}_i)^\times\right)} \approx I + \frac{h}{2}(\hat{\omega}_{i+1} + \hat{\omega}_i)^\times
\end{equation}
we have,
\begin{align}
    \Delta Q_i & = Q_i\left[\hat{R}_i\left(I + \frac{h}{2}(\hat{\omega}_{i+1} + \hat{\omega}_i)^\times\right)\hat{R}_i^T - I\right] \nonumber \\
    & = \frac{h}{2}Q_i\left(\hat{R}_i(\hat{\omega}_{i+1} + \hat{\omega}_i)^\times\hat{R}_i^T\right) \nonumber \\
    & = \frac{h}{2}Q_i\left(\hat{R}_i(\hat{\omega}_{i+1} + \hat{\omega}_i)\right)^\times .
\end{align}

In the absence of measurement errors, we have $\Tilde{U}_i^m = R_i^TE_i$. \\

Therefore,
\begin{align}
    \Delta\mathcal{U}_i & = -\frac{h}{2}\left\langle\,Q_i\left( \hat{R}_i \left(\omega_{i+1} + \omega_i \right) \right)^\times,K\right\rangle \nonumber \\
    & = -\frac{h}{2}\left \langle\,R_i(\omega_{i+1} + \omega_i)^\times\hat{R}_i^T,E_iW_iE_i^T\right\rangle \nonumber\\
    & = -\frac{h}{2}\left \langle\,(\omega_{i+1} + \omega_i)^\times\hat{R}_i^T,R_i^TE_iW_iE_i^T\right\rangle \nonumber\\
    & = -\frac{h}{2}\left \langle\,(\omega_{i+1} + \omega_i)^\times\hat{R}_i^T,\Tilde{U}_i^mW_iE_i^T\right\rangle
\end{align}

We have $L_i := E_iW_i(\Tilde{U}_i^m)^T$.
\begin{align}\label{eq:32}
    \Delta\mathcal{U} & = -\frac{h}{2}\left \langle\,(\omega_{i+1} + \omega_i)^\times,L_i^T\hat{R}_i\right\rangle \nonumber\\
    & = -\frac{h}{4}\left \langle\,(\omega_{i+1} + \omega_i)^\times,L_i^T\hat{R}_i - \hat{R}_i^TL_i\right\rangle \nonumber\\
    & = -\frac{h}{2}(\omega_{i+1}+\omega_i)^TS_{L_i}(\hat{R}_i)
\end{align}
where, $S_{L_i}(\hat{R}_i) = \text{vex}(L_i^T\hat{R}_i - \hat{R}_i^TL_i)$. Similarly we can compute the change in the kinetic energy as follows.
\begin{align}
        \Delta\mathcal{T} & = \mathcal{T}(\omega_{i+1}) - \mathcal{T}(\omega_i) \nonumber \\
        & = (\omega_{i+1} + \omega_i)^T\frac{m}{2}(\omega_{i+1} - \omega_i) \nonumber\\
       \Delta\mathcal{T} & = (\omega_{i+1} + \omega_i)^T\frac{m}{2}(\omega_{i+1} - \omega_i)
\end{align}

Therefore, the change in the value of the candidate Lyapunov function can be computed as, 
\begin{align}
    \Delta V_i & = V_{i+1} - V_i = \Delta\mathcal{T}_i + k_p\Delta\mathcal{U}_i \nonumber\\
    & = \frac{1}{2}\left(\omega_{i+1} + \omega_i\right)^T\left(m(\omega_{i+1} - \omega_i) - k_phS_{L_i}(\hat{R}_i)\right)
\end{align}

Now, for $\Delta V_i$ to be negative definite, 
\begin{equation}
    m(\omega_{i+1} - \omega_i) - k_phS_{L_i}(\hat{R}_i) = -l(\omega_{i+1} + \omega_i)
\end{equation}
where $l > 0, l \neq m$. Therefore,
\begin{equation}\label{eq:36}
    \omega_{i+1} = \frac{1}{m+l}\left[ (m-l)\omega_i + k_phS_{L_i}(\hat{R}_i) \right]
\end{equation}
and $\Delta V_i$ simplifies to,
\begin{equation}
    \Delta V_i = -\frac{l}{2}\left(\omega_{i+1} + \omega_i\right)^T\left(\omega_{i+1} + \omega_i\right).
\end{equation}

We employ the discrete-time La-Salle invariance principle from \cite{mei2017lasalle} considering our domain ($SO(3)\times\mathbb{R}^3$) to be a subset of $\mathbb{R}^{12}$, and for this we first compute $\mathscr{E} := \{(Q_i,\omega_i) \in SO(3)\times\mathbb{R}^3 | \Delta V_i(Q_i,\omega_i) = 0\} = \{(Q_i,\omega_i) \in SO(3)\times\mathbb{R}^3 \; | \; \omega_{i+1} + \omega_i = 0\}$. From \cref{eq:23},  $\omega_{i+1} + \omega_i = 0$ implies that,
\begin{equation}
    Q_{i+1} = Q_i
\end{equation}
Also, from \cref{eq:32} we have  $\Delta\mathcal{U} = 0$ whenever $\omega_{i+1} + \omega_i = 0$. This implies that the potential function, which is a Morse function according to \cref{lemma:1}, is not changing and therefore has converged to one of its stationary points. Stationary points of the Morse function are characterised by the solutions of,
\begin{equation}\label{eq:39}
    S_K(Q_i) = 0 \Rightarrow \text{vex}\left(KQ_i^T - Q_iK\right) = 0 \Rightarrow KQ_i^T = Q_iK.
\end{equation}

Multiplying \cref{eq:39} from the right hand side by $Q_i$  and from the left hand side by $Q_i^T$, and also noting that $Q_iQ_i^T = Q_i^TQ_i = I_{3\times3}$, we have the following relation at the critical points.
\begin{equation}\label{eq:40}
    Q_i^TKQ_i^TQ_i = Q_i^TQ_iKQ_i \Rightarrow Q_i^TK = KQ_i
\end{equation}

Now, $L_i = E_iW_i(\Tilde{U}_i^m)^T=E_iW_i(R_i^TE_i)^T= (E_iW_iE_i^T)R_i=KR_i$, which will further give us,
\begin{align}\label{eq:41}
    \left(S_{L_i}(\hat{R}_i)\right)^\times & = L_i^T\hat{R}_i - \hat{R}_i^TL_i \nonumber\\
    & = R_i^TK\hat{R}_i - \hat{R}_i^TKR_i
\end{align}

Multiplying \cref{eq:41} from the right hand side by $\hat{R}_i^T$  and from the left hand side by $\hat{R}_i$,
\begin{align}
    \hat{R}_i\left(S_{L_i}(\hat{R}_i)\right)^\times\hat{R}_i^T & = \hat{R}_iR_i^TK - KR_i\hat{R}_i^T \nonumber \\
    & = Q_i^TK - KQ_i
\end{align}

At the critical points from \cref{eq:39}, we have that $\hat{R}_i\left(S_{L_i}(\hat{R}_i)\right)^\times\hat{R}_i^T = 0$. Since both $\hat{R}_i$ and $\hat{R}_i^T$ are orthogonal matrices, the following will hold true at the critical points, 
\begin{equation}
   \left(S_{L_i}(\hat{R}_i)\right)^\times = 0 \Rightarrow S_{L_i}(\hat{R}_i) = 0
\end{equation}

Similarly, $S_{L_{i+1}}(\hat{R}_{i+1}) = 0$. Substituting this information in \cref{eq:36} yields,
\begin{equation}
    \omega_{i+1} = \frac{1}{m+l} (m-l)\left(\omega_i \right)
\end{equation}

Now if, $\omega_{i+1} + \omega_i = 0$, we have,
\begin{equation}
    \frac{2m}{m+l}\omega_i = 0 \Rightarrow \omega_i = 0 \Rightarrow \omega_i = \omega_{i+1} = 0
\end{equation}

We now evaluate the set to be $\mathscr{E} = \{(Q_i,\omega_i) \in SO(3)\times\mathbb{R}^3 \; | \; Q_i \in C_Q, \omega_i = 0\}$ further, recognising the fact that this is also an invariant set. Hence, we obtain $\mathscr{M} = \mathscr{E} = \{(Q_i,\omega_i) \in SO(3)\times\mathbb{R}^3 \; | \; Q_i \in C_Q, \omega_i = 0\}$. Furthermore, we have that $\mathscr{M} \subset V^{-1}(0)$.
Therefore, we obtain the positive limit set as the set,
\begin{align}
    \mathscr{I} & := \mathscr{M} \cap V^{-1}(0) \nonumber \\ & = \{(Q,\omega) \in SO(3)\times\mathbb{R}^3\; | \; Q \in C_Q, \omega = 0 \}
\end{align}

Therefore, in the absence of measurement errors, all the solutions of this filter
converge asymptotically to the set $\mathscr{I}$. More specifically, the attitude estimation error converges
to the set of critical points of $\langle\,I-Q,K\rangle$. The unique global minimum of this function is at $(Q,\omega) = (I,0)$ from \cref{lemma:2}, thus proving our claim of asymptotic stability. The remainder of this proof is similar to the last part of the proof of stability of the variational attitude estimator in \cite{izadi2014rigid}\\

Now consider the set,
\begin{equation}
    \mathscr{C} = \mathscr{I} \backslash (I,0)
\end{equation}
which consists of all the stationary states that the estimation errors may converge to,
besides the desired estimation error state $(I,0)$. Note that all states in the stable
manifold of a stationary state in $\mathscr{C}$ will converge to this stationary state. From
the properties of the critical points $Q^i \in C_Q \backslash (I)$ of $\mathit{\Phi}(\langle\,K,I-Q\rangle)$  given in \cref{lemma:2}.  we see that the stationary points in $\mathscr{I}\backslash(I,0) = \{(Q^i,0) : Q^i \in C_Q \backslash (I)\}$ have stable manifolds whose dimensions depend on the index of $Q^i$. Since the angular velocity estimate error $\omega$ converges globally to the zero vector, the dimension of the stable manifold $\mathcal{M}_i^S$ of $(Q^i,0) \in SO(3)\times\mathbb{R}^3$ is
\begin{equation}\label{eq:48}
    \text{dim}(\mathcal{M}_i^S) = 3 + (3 - \text{index of }\; Q^i) = 6 - \text{index of }\; Q^i
\end{equation}

Therefore, the stable manifolds of $(Q,\omega) = (Q^i,0)$ are three-dimensional, four
dimensional, or five-dimensional, depending on the index of $Q^i \in C_Q \backslash (I)$ according to \cref{eq:48}. Moreover, the value of the Lyapunov function $V(Q_i,\omega_i)$ is non
decreasing (increasing when $(Q_i,\omega_i) \notin \mathscr{I}$) for trajectories on these manifolds when going backwards in time. This implies that the metric distance between error
states $(Q,\omega)$ along these trajectories on the stable manifolds $\mathcal{M}_i^S$ grows with the time separation between these states, and this property does not depend on the
choice of the metric on $SO(3) \times \mathbb{R}^3$. Therefore, these stable manifolds are embedded (closed) sub-manifolds of $SO(3) \times \mathbb{R}^3$ and so is their union. Clearly, all states starting in the complement of this union, converge to the stable equilibrium
$(Q,\omega) = (I,0)$; therefore the domain of attraction of this equilibrium is,
\begin{equation}
    DOA{(I,0)} = SO(3)\times\mathbb{R}^3\backslash   \{\cup_{i=1}^3\mathcal{M}_i^S\}
\end{equation}
which is a dense open subset of $SO(3)\times\mathbb{R}^3$.
\end{proof}

\section{Numerical Simulations}\label{sec:4}

This section presents numerical simulation results of the discrete-time estimator presented in \cref{sec:3}. The estimator is simulated over a time interval of $T$ = 60 s, with a step-size of $h = 0.01 s$. The rigid body is assumed to have an initial attitude and angular velocity given by,
$$ R_0 = \text{expm}_{SO(3)}\left(\left( \frac{\pi}{4}\times\left[\frac{4}{7}, \; \frac{2}{7}, \; \frac{5}{7}\right]^T \right)^\times\right),$$ $$\text{and} \;\; \Omega_0 = \frac{\pi}{60}\times[-1.2, \; 2.1, \; -1.9]^T \; rad/s $$

The inertial scalar gain is $m = 100$ and the dissipation term is chosen to be $l = 40$. The difference of sampling rate between measurements of angular velocity and measurements inertial vectors in body-fixed frame is taken to be $n = 10$. Furthermore, the value of gain $k_p$ is chosen to be $k_p = 150$.  $W$ is selected based on the measured set of inertial vectors $E$ at each instant such that it satisfies \cref{lemma:1}. Initially estimated states have the following initial estimation errors:
$$ Q_0 = \text{expm}_{SO(3)}\left(\left( \frac{\pi}{2.5}\times\left[\frac{4}{7}, \; \frac{2}{7}, \; \frac{5}{7}\right]^T \right)^\times\right),$$ $$\text{and} \;\; \omega_0 = \frac{\pi}{60}\times[0.001, \; -0.002, \; 0.003]^T \; rad/s $$.

It has been assumed that there are at most 9 inertially known directions that are being measured by the sensors attached to the rigid body. The number of observed direction can vary randomly between 2 to 9 at each time instant. In the case where the number of observed directions is 2, the cross product of the two measurements is used as the third measurement. The standard rigid body dynamics are used to produce true states of the rigid body by applying sinusoidal forces. These true states are used to simulate the observed direction in the body-fixed frame, as well as compare true states and estimated states. Bounded, zero-mean random noises are generated which are then added to the real quantities in order to simulate real measurements. Based on coarse attitude sensors like sun sensors and magnetometers, a random noise bounded in magnitude by $2.4^\circ$ is added to the matrix $U = R^TE$ in order to generate measured $U^m$. Similarly, a random noise bounded in magnitude by $0.97^\circ/s$, which is close to real noise levels of coarse rate gyros, is added $\Omega$ to generate measured $\Omega_m$.
The principle angle $\phi$ of the rigid body's attitude estimation error $Q$ is shown in the \cref{phi_vs_t}. Components of estimation error $\omega$ in the rigid body's angular velocity are shown in \cref{omega_vs_t}. All the estimation errors are seen to converge to a bounded neighborhood of $(Q,\omega) = (I,0)$ with the bound being dictated by sensor noise magnitude bounds. The rate of convergence is dictated by the value of $k_p$. Increasing value of $k_p$ leads to faster convergence of estimation errors. However, the bound on errors in the presence of noise increases with the value of $k_p$. If the value of $l$ is closer to $m$, i.e. $m-l$ is smaller, then the bound on error decreases while increasing the time of convergence.
\begin{figure}
	\centering
	\includegraphics[scale = 0.42,trim={1.5cm 0cm 0.5cm 1cm},clip]{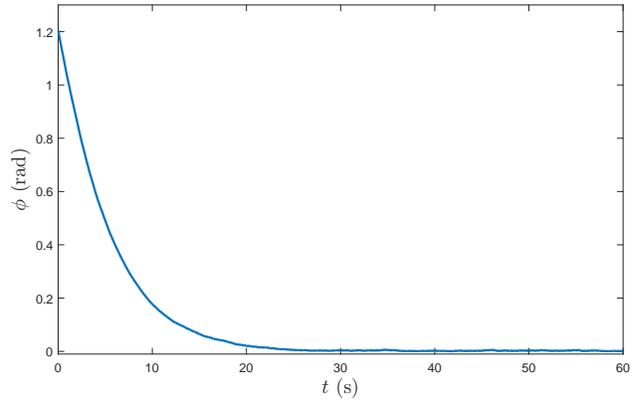}
	\caption{Principle angle of the attitude estimation error}
	\label{phi_vs_t}
\end{figure}

\begin{figure}
	\centering
	\includegraphics[scale = 0.42,trim={1cm 0 0.5cm 1cm},clip]{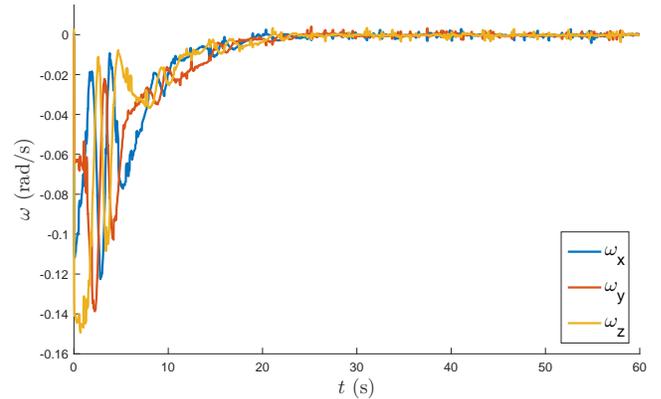}
	\caption{Angular velocity estimation error}
	\label{omega_vs_t}
\end{figure}

\section{Conclusion}\label{sec:5}

We develop a geometric attitude and angular velocity estimation scheme using discrete-time Lyapunov stability analysis in the presence of multi-rate measurements. The attitude determination problem from two or more vector measurements in the body-fixed frame is formulated as Wahba's optimization problem. To overcome the multi-rate challenge, a discrete-time model for attitude kinematics is used to propagate the inertial vector measurements forward in time. The filtering scheme is obtained with the aid of an appropriate discrete-time Lyapunov function consisting of Wahba's cost function as an artificial potential term and a kinetic energy-like term that is quadratic in the angular velocity estimation error. The filtering scheme was proven to be exponentially stable in the absence of measurement noise and the domain of convergence is proven to be almost global. Furthermore, the rate of convergence of the estimated states to the real state can be controlled by choosing appropriate gains. Numerical simulations were provided with realistic inputs in the presence of bounded measurement noise. Numerical simulations verified that the estimated states converge to a bounded neighborhood of $(I,0)$. Future endeavors are towards obtaining an optimal estimation multi-rate estimation scheme via variational methods, while also guaranteeing asymptotic stability of estimation errors.

\bibliographystyle{IEEEtran}
\bibliography{bibliography}

\end{document}